%% file: 12OCT_arXivSubmission_Cleaned.tex
\documentclass[conference,letterpaper]{IEEEtran}
\usepackage[top=0.73in,bottom=1.05in,left=0.665in,right=0.64in]{geometry}
\IEEEoverridecommandlockouts

\usepackage{enumerate}
\usepackage[utf8]{inputenc} 

\usepackage[T1]{fontenc}
\usepackage{hyperref} 
\usepackage{url}
\usepackage{ifthen}
\usepackage{tcolorbox}
\usepackage[noadjust]{cite}
\usepackage[cmex10]{amsmath} 
\usepackage{amssymb,amsfonts,amscd,amsbsy}
\usepackage{cleveref}
\usepackage{graphicx}
\usepackage{stmaryrd} 
\usepackage[mathscr]{eucal}
\usepackage{mathtools}
\usepackage{bbm}
\usepackage{xcolor}
\usepackage{xfrac}
\usepackage{algorithm}
\usepackage{algpseudocode}
\usepackage{enumerate}
\usepackage{enumitem}
\usepackage{multirow}

\usepackage{booktabs}
\usepackage{tikz} 
\usepackage{pgfplots} 
\pgfplotsset{compat=newest}
\usetikzlibrary{patterns}
\usepackage{bm}

\newtheorem{theorem}{Theorem}
\numberwithin{theorem}{section}
\newtheorem{corollary}{Corollary}
\numberwithin{corollary}{section}
\newtheorem{lemma}{Lemma}
\numberwithin{lemma}{section}
\newtheorem{definition}{Definition}
\numberwithin{definition}{section}

\newtheorem{example}{Example}
\newtheorem{remark}{Remark}
\newtheorem{proposition}{Proposition}

\usepackage{etoolbox}

\makeatletter
\renewcommand*{\thesection}{\arabic{section}}
\renewcommand*{\thesubsection}{\thesection.\arabic{subsection}}
\renewcommand*{\p@subsection}{}
\renewcommand*{\thesubsubsection}{\thesubsection.\arabic{subsubsection}}
\renewcommand*{\p@subsubsection}{}
\makeatother

\usepackage{titlesec}
\titleformat{\subsection}[runin]
  {\normalfont\normalsize\bfseries}{\thesubsection.}{3pt}{}[.]
\titleformat{\subsubsection}[runin]
  {\normalfont\normalsize\itshape}{\thesubsubsection.}{3pt}{}[\mbox{~}]

\input{macros.tex}

\DeclareMathOperator{\LCM}{\mathsf{LCM}}
\newcommand{\Pn}{\mathcal{P}_n}
\newcommand{\Dn}{\Delta_n}

\newcommand{\Ha}{H_\alpha}
\newcommand{\da}{d_\alpha}
\newcommand{\un}{\mathsf{u_n}}
\newcommand{\eone}{\mathsf{e_1}}
\newcommand{\TV}{\mathrm{TV}}

\newcommand{\KS}{\mathrm{KS}}
\newcommand{\dB}{d_B}

\newcommand{\Unif}{\mathrm{Unif}}
\newcommand{\KL}[2]{D\!\left(#1\,\middle\|\,#2\right)}
\newcommand{\sort}{\mathsf{sort}(}
\newcommand{\kpq}{k_{p,q}}
\newcommand{\Ch}{\mathsf{Ch}}

\title{On Measures of Dissimilarity ~\\ on the Majorization Simplex}
\author{Tarun Hariharan, Prasad Krishnan, Abhishek Deshpande 

\thanks{Tarun, Prasad and Abhishek are at the International Institute of Information Technology, Hyderabad, 500032, India (email: $\{$tarun.hariharan@research., prasad.krishnan@, abhishek.deshpande@$\}$iiit.ac.in). 
}
}

\begin{document}

\maketitle

\begin{abstract}
The majorization simplex $\Pn$ is the set of probability vectors of length $n$ whose entries are arranged in non-increasing order. Under the majorization order, $\Pn$ is a lattice, and it is the natural space for comparing distributions when the labels of the outcomes are not relevant. In this work, we study measures of dissimilarity (and, similarity) on $\Pn$. Our first set of results concerns the R\'enyi-entropy based distance $d_\alpha(p,q)=H_\alpha(p)+H_\alpha(q)-2H_\alpha(p\vee q)$ proposed by Yadav and Shkel (2026), which generalizes a metric of Cicalese, Gargano and Vaccaro (2013). We show that $d_\alpha$ is a metric on $\Pn$ for $\alpha\in[1,\infty)$, while for $\alpha\in[0,1)\cup\{\infty\}$, $\da$ is not a metric in general. Our main result in this part identifies the necessary and sufficient conditions for equality in the triangle inequality of $d_\alpha$ when $\alpha>1$, for which we identify the equality conditions for the supermodularity of the R\'enyi entropy $H_\alpha$. Our second set of results concerns the total variation distance and the Hellinger affinity (also called the Bhattacharyya coefficient) on $\Pn$. We obtain a number of bounds on these measures, along with some of their equality conditions. The resulting Hellinger affinity based upper bound on the total variation distance can be smaller than the bounds obtained from Pinsker's and the Bretagnolle--Huber inequalities.
\end{abstract}




\section{Introduction}
\label{sec:intro}

Comparing probability distributions is a foundational task in information theory, statistical mechanics, and data science. Traditionally, this is achieved using $f$-divergences, such as the Kullback--Leibler (KL) divergence \cite{CoverThomas2006,Pinsker1964}. However, not all $f$-divergences are metrics or pseudo-metrics, which makes such divergences unsuitable for formal distance geometry. Furthermore, classical divergences operate on the labeled probability simplex $\Delta_n$, which inherently assumes that the specific identity or label of each underlying state is fixed and meaningful \cite{Leinster21}.

In a wide array of theoretical and applied disciplines, however, the specific labels of outcomes are irrelevant. In financial economics, metrics on sorted distributions are used to quantify wealth inequality and anonymity profiles \cite{The67, CGV13}. In theoretical ecology, they are used for measuring structural biodiversity, allowing researchers to compare the rank-abundance profiles of coral reefs \cite{Wong2018CoralDiversity}, sponge microbiomes \cite{MoitinhoSilva2017}, and neotropical butterfly populations \cite{DeVriesWalla2001} independent of species identity \cite{Jost2006, LC12, Leinster21}. Furthermore, in quantum information theory, comparing sorted eigenspectra is fundamental to characterizing entanglement transformations, pure-state resource theories, and majorization flows \cite{HD19, Nielsen1999, Vidal1999, Desideetal2024}. 

By sorting probabilities in descending order, we abstract away the labels and isolate the distribution's intrinsic concentration or evenness \cite{MR23, GR24}. The space of such sorted distributions is equipped with a natural partial order known as majorization ($\preceq$), which mathematically formalizes the concept of one distribution being ``more concentrated'' than another \cite{MOA2011}. Notably, the space of sorted distributions under majorization forms a complete lattice, termed the \emph{majorization lattice} \cite{Bapat91}, where every pair of distributions possesses a uniquely defined supremum (the join $\vee$) and infimum (the meet, denoted by $\wedge$). Metrics and pseudo-metrics on this lattice are crucial for bounding the physical cost of transforming one quantum resource into another \cite{HD19}, or robustly measuring the distance between data profiles without label-matching overhead \cite{CGV13}. 

The study of metrics and pseudometrics on the majorization lattice have been progressing in the recent two decades. Cicalese and Vaccaro originally demonstrated that the Shannon entropy exhibits supermodularity on the majorization lattice \cite{CicaleseVaccaro2002}, a property subsequently exploited by Cicalese, Gargano, and Vaccaro (CGV) to define a specific structural metric on this space \cite{CGV13}. Recently, Yadav and Shkel proposed a generalization to the broader R\'enyi entropy family \cite{YS26}. In \cite{YS26}, the regimes of $\alpha$ for which the R\'enyi entropy $H_\alpha$ is subadditive and supermodular on the majorization lattice were further characterized, and similar analysis was extended to Tsallis entropy. In \cite{SJDC26}, the authors present some majorization `precursor' relations, based on which supermodularity and subadditivity relationships of a large class of entropies including Shannon, Tsallis, and R\'enyi entropies are argued. Large sub-families of the Sharma--Mittal family of entropies also satisfy subadditivity and supermodularity on the majorization lattice \cite{BV26}. Majorization theory has also been used in the theory of minimum-entropy couplings \cite{CGV_min-entropy-couplingTIT2019,SY_2023_ISIT_InfoSpectrumConverse_Coupling, compton_2022ISIT_greedyminentrcoupl} and functional representations of random variables \cite{YadShkel_2025ISIT_MinRenyiFunctionalRepres, serhat-shkel-2026ISIT-perfectfunctrepresentations, LiCheuk_2025_ISIT_DiscreteLayeredEntropy}.  
 Despite these advances, the complete collection of properties of R\'enyi metrics on the majorization lattice has not been sufficiently characterized. Further, the behaviour of classical distances (like total variation \cite{BH1979}) when confined to the majorization simplex is also not fully understood. 
 
 In this work, we expand the study of (dis)similarity measures on the majorization lattice, addressing these two gaps. We present two sets of results. Our first collection of results is regarding the metric $\da$ introduced in \cite{YS26}, which is defined as $d_\alpha(p,q) = H_\alpha(p) + H_\alpha(q) - 2H_\alpha(p\vee q)$, where $p,q$ are distributions in the $n$-ary majorization simplex $\Pn$. Our second collection of results concerns the classical (dis)similarity measures, in particular the Hellinger affinity (also known as the Bhattacharyya coefficient) and the total variation distance, leading to new bounds on these quantities on the simplex, which are complementary to existing bounds. 

Specifically, our primary contributions are organized as follows:
\begin{itemize}
    \item \textit{Characterization of the R\'enyi-Lattice Metric:} After a brief review of relevant conceptual foundations in Section \ref{sec:setup}, in Section \ref{sec:defndalpha_structural_properties} we give a proof that the R\'enyi lattice distance $d_\alpha(p,q)$ is a valid metric on the majorization simplex $\Pn$ for $\alpha \in [1, \infty)$, decaying into a pseudo-metric at the boundary limits. We also present some structural properties of $d_\alpha$, including bounds on the same and their equality conditions. Importantly, we present the necessary and sufficient conditions on the distributions in $\Pn$, required for equality in the triangle inequality of $d_\alpha$ and supermodularity of $H_\alpha$ ($\alpha > 1$). The identification of these precise conditions is novel to this present work, to the best of our knowledge. 
    \item \textit{Bounds on TV distance and Hellinger affinity on $\Pn$:} In Section \ref{sec:tv-kl}, we obtain new bounds on the total variation distance and Hellinger affinity on the majorization simplex, by carefully utilizing the ordering of the distributions involved. These new relationships complement existing bounds (including those that use the Kullback-Leibler divergence), and for some conditions are improvements on the same. 

\end{itemize}

\emph{Notation:} Throughout, $n\geq 2$ is fixed and all logarithms are natural, so all entropies are measured in nats, and we write $[n]:=\{1,\ldots,n\}$. For $a$ real, we use $a^+:=\max\{a,0\}$.   The notation $\Dn$ denotes the labeled probability simplex, $\Pn$ denotes the majorization simplex (see Definition \ref{def:labelledandsortedsimplices}), $\preceq$ denotes the majorization order on $\Pn$, $\wedge$ and $\vee$ denote the lattice meet and join, $\un=(1/n,\ldots,1/n)$  denotes the uniform distribution, and we denote $\eone=(1,0,\ldots,0)\in \Pn$. For $p\in\Pn$ and $A\subseteq[n]$, we write $p(A):=\sum_{i\in A}p_i$, and $\supp(p)$ for the support of $p$. For a pair $p,q\in\Pn$ we set $\kpq:=\max\{|\supp(p)|,|\supp(q)|\}\leq n$, and $\Unif(k)\in\Pn$ denotes the vector whose first $k$ entries equal $1/k$. 

\section{Preliminaries: Majorization and R\'enyi Entropy}
\label{sec:setup}

In this section, we establish the notation for probability simplices, recall the concept of majorization, and review the relevant properties of R\'enyi entropy on the majorization simplex. 

\begin{definition}[{Labeled and sorted simplices}]
\label{def:labelledandsortedsimplices}
The \emph{labeled simplex} is $\Dn=\{p=(p_1,\ldots,p_n):p_i\geq 0,\sum_{i=1}^n p_i=1\}$. The symmetric group $\mathsf{S}_n$ acts on $\Dn$ by permuting coordinates. The \emph{majorization simplex} \cite{CGV13, YS26}, denoted $\Pn$, is the set of \emph{sorted} probability vectors:
\begin{align}
\label{eqn:pnsimplex}
  \Pn = \{(p_1,\hdots,p_n)\in\Delta_n:\ p_1\geq\cdots\geq p_n\}.
\end{align}

Every $\mathsf{S}_n$-orbit in $\Dn$ possesses exactly one sorted representative, establishing a natural bijection between $\Pn$ and the quotient space $\Dn/\mathsf{S}_n$.
\end{definition}
 
For $p\in \Pn$, let $S_p(k) = \sum_{i=1}^k p_i,\forall k\in\{0,\dots,n\}$ (where we define $S_p(0)=0$ by convention), denote the discrete cumulative prefix sums. Note that $S_p(k), \forall k\in[n]$ uniquely characterizes $p$. We now recall a natural partial order that can be imposed on $\Pn$ using this quantity \cite{HLP1934, MOA2011}.
\begin{definition}[{Majorization \cite[Chapter 1, Sec. A]{MOA2011}}]
\label{def:majorization}
For $p,q\in\Pn$, $p$ is said to be \emph{majorized by} $q$, denoted $p\preceq q$, if their prefix sums satisfy 
\[
  S_p(k) \leq S_q(k), \qquad \forall k\in [n].
\] 
\end{definition}

Majorization extends to vectors of any length. For $u\in\mathbb R^N$, let $\sort u)$ denote its non-increasing rearrangement, and for $u\in\mathbb R^a$, $v\in\mathbb R^b$ let $\mathsf{sort}(u\,\|\,v)\in\mathbb R^{a+b}$ denote the non-increasing rearrangement of the concatenation of $u$ and $v$. For $a,b\in\mathbb R^N$ we write $a\preceq b$ if $\sum_i a_i=\sum_i b_i$ and $\sum_{i\le k}\sort a)_i\le\sum_{i\le k}\sort b)_i$ for all $k\in[N]$ \cite[Chapter 1, Sec. A]{MOA2011}. For vectors in $\Pn$, which are already sorted and have sum $1$, this is exactly Definition~\ref{def:majorization}. We use the same symbol $\preceq$ in both settings. 

It is also true \cite{Bapat91} that the partially ordered set $(\Pn,\preceq)$ is a \textit{complete lattice}, meaning every subset has a supremum and an infimum. For any two $p,q\in \Pn$, the \textit{meet} $p\wedge q$ and \textit{join} $p\vee q$ are uniquely characterized via their prefix sums:
\begin{align*}
  S_{p\wedge q}(k) &= \min\{S_p(k),S_q(k)\}\\
  S_{p\vee q} &= \LCM\bigl(\max\{S_p, S_q\}\bigr),
\end{align*}
where $\LCM$ denotes the \textit{least concave majorant}\footnote{The least concave majorant of $M$ on $\{0,\ldots,n\}$ is the pointwise smallest $J\ge M$ whose increments $J(k)-J(k-1)$ are non-increasing. It satisfies $J(0)=M(0)$ and $J(n)=M(n)$. An explicit iterative procedure to compute $p\vee q$, which is characterized by $S_{p\vee q}(k)$, is given in \cite[Sec.~III]{CicaleseVaccaro2002}.} over the discrete domain $\{0,\ldots,n\}$. 

It is easy to see that the minimum element of this lattice is the uniform distribution $\mathsf{u_n}=(1/n,\ldots,1/n)$, and the maximum is $\eone=(1,0,\ldots,0)$, the point mass on the first coordinate. 

For $\alpha \in (0,1) \cup (1,\infty)$ and $p\in\Pn$, the R\'enyi entropy is defined as
\[
  H_\alpha(p) = \frac{1}{1-\alpha}\log\Bigl(\sum_{i=1}^n p_i^\alpha\Bigr),
\]
extended by continuity to $H_1(p)=-\sum_{i=1}^n p_i\log p_i$, $H_0(p)=\log|\supp(p)|$, and $H_\infty(p)=-\log p_1$. \
\begin{definition}[Schur-Concavity]
\label{def:schur}
A function $f:\Pn\to\mathbb R$ is \emph{Schur-concave} if $p\preceq q\implies f(p)\geq f(q)$, and \emph{strictly} Schur-concave if the inequality is strict whenever $p\neq q$. 
\end{definition}

For every $\alpha\in[0,\infty]$, $H_\alpha$ is Schur-concave \cite[Chapter 3, Sec. A]{MOA2011}. For any finite $\alpha\in(0,\infty)$, $H_\alpha$ is strictly Schur-concave \cite{YS26}. Strict Schur-concavity fails at the boundary cases $\alpha\in\{0,\infty\}$.

Some results in the subsequent sections rely on the supermodularity of $H_\alpha$ over the lattice operations in $\Pn$, given as follows.
\begin{theorem}[{Supermodularity  of $H_\alpha$ in $\Pn$ \cite{CicaleseVaccaro2002, YS26}}]  
\label{thm:supermod}
For every $\alpha\in\{0\}\cup[1,\infty]$, R\'enyi entropy is supermodular:
\[
  H_\alpha(p) + H_\alpha(q) \leq H_\alpha(p\wedge q) + H_\alpha(p\vee q).
\]
\end{theorem}

For $\alpha\in(0,1)$, neither supermodularity nor submodularity hold \cite[Theorem 4]{YS26}.

\section{A R\'enyi-entropy based distance measure on $\Pn$ and its structural properties}
\label{sec:defndalpha_structural_properties}

In this section, we recall a R\'enyi entropy based distance metric on $\Pn$ that was introduced by Yadav and Shkel \cite{YS26}, termed $\da$. In \cite{CGV13},  Cicalese, Gargano and Vaccaro established a number of structural facts about the special case of $d_1$  beyond the metric axioms themselves. We obtain some properties of $\da$, for $\alpha\in[1,\infty)$. In particular, we show in Subsection \ref{subsec:conditionsforequalitysupermod} the necessary and sufficient conditions under which the triangle inequality for $\da$ holds with equality. For establishing these, we also obtain the equality conditions for supermodularity of the R\'enyi entropy. We start with the definition of $d_\alpha$. 
\begin{definition}[The R\'enyi metric $d_\alpha$ as proposed by \cite{YS26}]
\label{def:dalpha}
For $\alpha\in[1,\infty)$ and $p,q\in\Pn$,
\[
  d_\alpha(p,q) \triangleq H_\alpha(p) + H_\alpha(q) - 2H_\alpha(p\vee q).
\]
\end{definition}
At $\alpha=1$ this is exactly the metric $d$ of CGV~\cite[Definition~3]{CGV13}. 

 It is observed in \cite{YS26} that $d_\alpha$ is a metric\footnote{More generally, any Schur-concave, supermodular $F:\Pn\to\mathbb R$ induces a pseudo-metric (satisfying all properties of a metric except that $d_F(p,q)=0\implies p=q$) $d_F(p,q):=F(p)+F(q)-2F(p\vee q)$ on the majorization lattice, and $\da:=d_{H_\alpha}$ is the special case studied here.}, but we provide a proof here for the sake of completeness.

\begin{theorem}[$d_\alpha$ is a metric \cite{YS26}]
\label{thm:metric}
For each $\alpha\in[1,\infty)$, $(\Pn,d_\alpha)$ is a metric space. 
\end{theorem}

\begin{IEEEproof}
See Appendix~\ref{app:metric-proof}.
\end{IEEEproof}

In Appendix \ref{app:da-metric-range}, we discuss the cases of $\alpha\in[0,1)\cup\{\infty\}$ and show why $\da$ is not a metric for these, in the case of general $n$.

\subsection{Structural properties of $d_\alpha$}
The following proposition provides upper and lower bounds on $d_\alpha(p,q)$, and the conditions under which they hold with equality. Further, it also establishes that for any chain of comparable distributions $p,q,r \in \Pn$, $d_\alpha$ follows a certain additive relationship.
\begin{proposition}
\label{prop:lb-ub-chain}
Let $p,q,r\in\Pn$ and $\alpha \in [1, \infty)$. Then,
\begin{enumerate}
    \item \label{itm: lb}
    $d_\alpha(p,q)\geq |H_\alpha(p)-H_\alpha(q)|$ with equality if and only if $p\preceq q$ or $q\preceq p$.
    \item \label{itm: ub}
    $d_\alpha(p,q)\leq H_\alpha(p)+H_\alpha(q)$, with equality if and only if $p\vee q=\eone$, i.e.\ if and only if $p=\eone$ or $q=\eone$. 
    \item \label{itm: chain} 
    If $p\preceq r\preceq q$, then $d_\alpha(p,q)=d_\alpha(p,r)+d_\alpha(r,q)$. 
\end{enumerate}
\end{proposition}

\begin{IEEEproof}
Consider $p,q \in \Pn$. Assume, without loss of generality, that $H_\alpha(p)\geq H_\alpha(q)$. Then, we have 
\begin{align}
  d_\alpha(p,q) &= H_\alpha(p)+H_\alpha(q)-2H_\alpha(p\vee q)
  \\
  \label{eqn:dalphalowerbound} 
  &\geq  
  H_\alpha(p)+H_\alpha(q)-2H_\alpha(q) = H_\alpha(p)-H_\alpha(q), 
\end{align}
where $H_\alpha(p\vee q)\leq H_\alpha(q)$ holds by Schur-concavity. Now, if $p,q$ are not comparable to each other (i.e., $p \not\preceq q$ and $q \not\preceq p$), then, $p\vee q$ strictly majorizes $p,q$, so strict Schur-concavity gives $H_\alpha(p\vee q)<\min(H_\alpha(p),H_\alpha(q))$, which gives a strict inequality in \eqref{eqn:dalphalowerbound}. On the other hand, if without loss of generality $p\preceq q$, then $p\vee q=q$. Hence, we have equality in \eqref{eqn:dalphalowerbound}. This concludes the proof of (\ref{itm: lb}). 

As, $H_\alpha(p\vee q)\geq0$, (\ref{itm: ub}) follows immediately. Clearly, equality holds if and only if $H_\alpha(p\vee q)=0$, which holds if and only if $p\vee q=\eone$. This happens if and only if $p=\eone$ or $q=\eone$, thus concluding the proof of (\ref{itm: ub}).

    Since $p,q,r$ are comparable in (\ref{itm: chain}), by applying the $\da(p,q) = \Ha(p)-\Ha(q)$ identity to the pairs $(p,r)$, $(p,q)$ and $(r,q)$ and expanding, one can observe (\ref{itm: chain}).
\end{IEEEproof}

We have the following corollary.
\begin{corollary}
\label{corr:join-path}
    For any $p,q\in\Pn$, we have $d_\alpha(p,q)=d_\alpha(p,p\vee q)+d_\alpha(q,p\vee q)$.
\end{corollary}

\begin{IEEEproof}
    For $z\succeq p,q$ we have $z\succeq p\vee q$, so $d_\alpha(p,z)+d_\alpha(z,q)=H_\alpha(p)+H_\alpha(q)-2H_\alpha(z)\ge d_\alpha(p,q)$, with equality at $z=p\vee q$.
\end{IEEEproof}

Corollary \ref{corr:join-path} is the R\'enyi analogue of \cite[End of Section 3]{CGV13}  variational characterization $d(p,q)=\min_{z\succeq p,q}[d(p,z)+d(z,q)]$, with the minimum attained at $z=p\vee q$.

While $H_\alpha(p)$ is non-increasing in $\alpha$, the following example shows that $\alpha\mapsto d_\alpha(p,q)$ is generally not monotonic on $[1,\infty)$. 
\begin{example}
\label{ex:nonmonotone}
Let $p=(2/3,1/3)$ and $q=(5/6,1/6)$ in $\mathcal{P}_2$. Since $p\preceq q$, we have $p\vee q=q$, and hence, for $\alpha>1$,
\[
  d_\alpha(p,q)=H_\alpha(p)-H_\alpha(q)=\frac{1}{\alpha-1}\log\frac{5^\alpha+1}{4^\alpha+2^\alpha}.
\]
In particular, $d_2(p,q)=\log(13/10)$, $d_3(p,q)=(1/2)\log(7/4)$, and $d_4(p,q)=(1/3)\log(313/136)$. Since $(13/10)^2=169/100<175/100=7/4$ and $(7/4)^3=343/64>97969/18496=(313/136)^2$, we obtain $d_2(p,q)<d_3(p,q)>d_4(p,q)$, so $\alpha\mapsto d_\alpha(p,q)$ is not monotone on $[1,\infty)$. 
\end{example}

\subsection{Equality in the triangle inequality for $d_\alpha$}
\label{subsec:conditionsforequalitysupermod}
For $a,b\in\Pn$ with $a\preceq b$, let $[a,b]:=\{s\in\Pn:a\preceq s\preceq b\}$ denote the lattice interval between them. The main result of this subsection characterizes equality in the triangle inequality of Theorem~\ref{thm:metric}. Specifically, we prove the following result. 

\begin{theorem}[Equality in the triangle inequality]
\label{thm:tri-eq}
Let $\alpha\in(1,\infty)$ and $p,q,r\in\Pn$. Then the triangle inequality for $d_\alpha$ holds with equality, i.e., $d_\alpha(p,q)=d_\alpha(p,r)+d_\alpha(r,q)$ if and only if
\[
  p\preceq r\preceq p\vee q \qquad\text{or}\qquad q\preceq r\preceq p\vee q,
\]
that is, if and only if $r\in[p,p\vee q]\cup[q,p\vee q]$. The ``if'' direction holds for every $\alpha\in[1,\infty)$.
\end{theorem}

We prove Theorem~\ref{thm:tri-eq} in three steps. Lemma~\ref{lem:comparability} identifies when supermodularity of R\'enyi entropy holds with equality; this is the only place where $\alpha>1$ is needed. Lemma~\ref{lem:tri-eq-conditions} reduces equality in the triangle inequality to three conditions on the lattice, one of which is tightness of supermodularity. The two are combined at the end. We now proceed through these steps. 

The following lemma identifies a condition, universal in $\alpha\in(1,\infty)$, for equality in supermodularity of R\'enyi entropy.

\begin{lemma}[Equality in supermodularity of R\'enyi entropy]
\label{lem:comparability}
Let $\alpha\in(1,\infty)$ and $x,y\in\Pn$. Then
\begin{align}
\label{eqn:supermodequality}
  H_\alpha(x)+H_\alpha(y) = H_\alpha(x\wedge y)+H_\alpha(x\vee y)
\end{align}
if and only if $x$ and $y$ are comparable (i.e.\ $x\preceq y$ or $y\preceq x$).
\end{lemma}

\begin{IEEEproof}
($\Leftarrow$) Without loss of generality, suppose $x\preceq y$. Then $x\wedge y=x$ and $x\vee y=y$, and the two sides of \eqref{eqn:supermodequality} coincide term-wise. The case $y\preceq x$ is identical.

($\Rightarrow$) This direction is proved in Appendix~\ref{app:comparability-proof}.

\end{IEEEproof}
We now prove a lemma which captures the equality in the triangle inequality in three conditions. 
\begin{lemma}[Equality in the triangle inequality: lattice conditions]
\label{lem:tri-eq-conditions}
Let $\alpha\in[1,\infty)$, $p,q,r\in\Pn$, $x:=p\vee r$ and $y:=r\vee q$. Then $d_\alpha(p,q)=d_\alpha(p,r)+d_\alpha(r,q)$ if and only if
\begin{enumerate}
  \item[\normalfont(C1)] $x\wedge y = r$;
  \item[\normalfont(C2)] $x\vee y = p\vee q$;
  \item[\normalfont(C3)] $H_\alpha(x)+H_\alpha(y)=H_\alpha(x\wedge y)+H_\alpha(x\vee y)$, i.e.\ supermodularity (\Cref{thm:supermod}) is tight for the pair $(x,y)$.
\end{enumerate}
\end{lemma}

\begin{IEEEproof}
Expanding the definition of $d_\alpha$ and rearranging,
\begin{align}
\nonumber\tfrac12&\bigl[d_\alpha(p,r)+d_\alpha(r,q)-d_\alpha(p,q)\bigr]
  \\
  \nonumber
  &= H_\alpha(r)+H_\alpha(p\vee q)-H_\alpha(x)-H_\alpha(y)\\
  \nonumber
  &= \bigl[H_\alpha(r)-H_\alpha(x\wedge y)\bigr]
   + \bigl[H_\alpha(p\vee q)-H_\alpha(x\vee y)\bigr]\\
  \label{eqn:proofof_tri:eq}
  &\quad + \bigl[H_\alpha(x\wedge y)+H_\alpha(x\vee y)-H_\alpha(x)-H_\alpha(y)\bigr].
\end{align}
Since $r\preceq x$ and $r\preceq y$, we have $r\preceq x\wedge y$; since $p\preceq x\preceq x\vee y$ and $q\preceq y\preceq x\vee y$, we have $p\vee q\preceq x\vee y$. By Schur-concavity of $H_\alpha$, the first two brackets in \eqref{eqn:proofof_tri:eq} are non-negative, and by supermodularity (\Cref{thm:supermod}) so is the third. Hence the left-hand side of \eqref{eqn:proofof_tri:eq} is non-negative, which is the triangle inequality, and it vanishes if and only if each bracket vanishes. By strict Schur-concavity of $H_\alpha$, the first bracket vanishes if and only if $r=x\wedge y$, which is (C1), and the second vanishes if and only if $p\vee q=x\vee y$, which is (C2). The third vanishes if and only if (C3) holds.
\end{IEEEproof}

We now complete the proof of Theorem~\ref{thm:tri-eq}, using Lemmas \ref{lem:comparability} and \ref{lem:tri-eq-conditions}. 
\begin{IEEEproof}[Proof of Theorem~\ref{thm:tri-eq}]
Let $x:=p\vee r$ and $y:=r\vee q$.

\emph{Only if part:} Suppose equality holds in the triangle inequality. By Lemma~\ref{lem:tri-eq-conditions}, (C1)--(C3) hold. By Lemma~\ref{lem:comparability} (which uses $\alpha>1$), (C3) says that $x$ and $y$ are comparable. Suppose $x\preceq y$. Then $x\wedge y=x$ and $x\vee y=y$, so (C1) reads $p\vee r=r$, i.e.\ $p\preceq r$, and (C2) reads $y=p\vee q$. Since $r\preceq r\vee q=y=p\vee q$, we conclude $p\preceq r\preceq p\vee q$. 

Note that, in the case $y\preceq x$, (C1) gives $r\vee q=r$, i.e.\ $q\preceq r$, and (C2) gives $x=p\vee q$, so $r\preceq p\vee r=p\vee q$; hence $q\preceq r\preceq p\vee q$. This concludes the proof of the only if part.

\emph{If part:} Suppose $p\preceq r\preceq p\vee q$; the case $q\preceq r\preceq p\vee q$ follows by exchanging $p$ and $q$, since both sides of the claim are symmetric. Then $x=p\vee r=r$. 

Since $p\preceq r\preceq r\vee q$ and $q\preceq r\vee q$, the element $r\vee q=y$ is an upper bound of $p$ and $q$, hence $p\vee q\preceq y$. Conversely, $r\preceq p\vee q$ and $q\preceq p\vee q$ give $y=r\vee q\preceq p\vee q$. Thus $y=p\vee q$.

From arguments in the previous two paragraphs, we have $x=r\preceq p\vee q=y$. Consequently $x\wedge y=x=r$ and $x\vee y=y=p\vee q$, which are (C1) and (C2). Further, (C3) holds because  $x\preceq y$. Thus, by Lemma~\ref{lem:tri-eq-conditions},  we have equality in the triangle inequality. Note that, for this direction (`if part'), the proof uses no property of $\alpha$ beyond $\alpha\in[1,\infty)$. This completes the proof of the `if part' and that of the theorem. 
\end{IEEEproof}

\begin{remark}
Theorem~\ref{thm:tri-eq} subsumes Corollary \ref{corr:join-path} and item (3) in Proposition \ref{prop:lb-ub-chain}. Indeed, taking $r=p\vee q$ in Theorem~\ref{thm:tri-eq} gives Corollary~\ref{corr:join-path}, since $p\preceq p\vee q\preceq p\vee q$. Similarly, if $p\preceq r\preceq q$, then $r\preceq q\preceq p\vee q$, so $r\in[p,p\vee q]$. This recovers item (3) of Proposition \ref{prop:lb-ub-chain}.
\end{remark}

By \Cref{thm:tri-eq}, equality in the triangle inequality forces $p\wedge q\preceq r\preceq p\vee q$: the upper bound is $r\preceq x\preceq x\vee y=p\vee q$ by (C2) from the proof of that theorem, and the lower bound follows from (C1), since $x=p\vee r\succeq p$ and $y=r\vee q\succeq q$ give $r=x\wedge y\succeq p\wedge q$. A reader may thus expect that equality in the triangle inequality holds for every $r\in[p\wedge q,\,p\vee q]$. Theorem~\ref{thm:tri-eq} shows that this range for $r$ is $r\in[p,p\vee q]\cup[q,p\vee q]$, which is a \textit{proper} subset of $[p\wedge q,p\vee q]$ whenever $p$ and $q$ are incomparable. The following example makes the gap concrete.

\begin{example}[Illustrative example for Theorem \ref{thm:tri-eq}]
\label{ex:n3-worked}

Let $p=(0.5,0.3,0.2)$ and $q=(0.45,0.45,0.1)$, so that $p\wedge q=(0.45,0.35,0.2)$ and $p\vee q=(0.5,0.4,0.1)$. By Theorem~\ref{thm:tri-eq}, equality in the triangle inequality for $\alpha>1$ holds exactly for $r$ in the two intervals $[p,p\vee q]$ and $[q,p\vee q]$, which meet only at $p\vee q$. For $\alpha=2$, the choice $r=(0.5,0.35,0.15)\in[p,p\vee q]$ gives $d_2(p,r)+d_2(r,q)=d_2(p,q)\approx0.1121$. In contrast, $r=(0.48,0.37,0.15)$ lies in the interval $[p\wedge,p\vee q]$ but in neither interval (neither $p\preceq r$ nor $q\preceq r$), and gives $d_2(p,r)+d_2(r,q)\approx0.1233>d_2(p,q)$.
\end{example}

\subsection{Relationship to other known quantities in econometrics and ecology}
\label{subsec:relationshiptoother_diversityquantities}
 The order-$\alpha$ Hill number is defined as $D_\alpha(p)\triangleq\exp H_\alpha(p)$ \cite{Hill73}, which has been used extensively as an effective number of species in the study of ecological diversity \cite{Jost2006, LC12, Leinster21}, for instance for comparing coral reef \cite{Wong2018CoralDiversity}, sponge microbiome \cite{MoitinhoSilva2017} and butterfly \cite{DeVriesWalla2001} communities. The order-$2$ Hill number, $D_2(p)=(\sum_i p_i^2)^{-1}$, is also popularly referred to as the inverse-Simpson index \cite{Simpson49}. Clearly, we have $$d_\alpha(p,q)=\log\bigl(D_\alpha(p)D_\alpha(q)\bigr)-2\log D_\alpha(p\vee q).$$ Thus, $d_\alpha$ is always a log-ratio comparison of Hill numbers. 

The Theil's inequality index~\cite{The67} is given by $\log n-H(p)$, it has been observed~\cite{CGV13} that the Theil
index is the distance from $p$ to the bottom of the lattice $\Pn$. Section VI of \cite{YS26} proposes the \emph{R\'enyi--Theil index} for $\alpha \in [1,\infty) $: $\log n-H_\alpha(p)$, where $\alpha$ dictates the degree to which deviation from uniformity is studied.
 
$d_\alpha(p,\cdot)$ evaluated at the two extremes of the lattice, $\un$ and $\eone$, recovers, respectively, the R\'enyi--Theil deviation from uniformity and the R\'enyi entropy itself.

\section{Bounding Total Variation Distance and Hellinger Affinity on the Majorization Simplex}

\label{sec:tv-kl}

The preceding sections studied the geometry of $\Pn$ through the entropy-based metric $\da$. We now turn to two classical measures, the total variation distance (a dissimilarity metric) and Hellinger affinity (a similarity measure), and examine their behaviour and properties when restricted to sorted probability vectors, i.e., $p,q\in\Pn$.

For $p,q\in\Pn$, the total variation distance $\TV(p,q)=\frac12\sum_{i=1}^n|p_i-q_i|$ has the following equivalent forms. The sup-over-events form is
\begin{align}
  \TV(p,q)&=\max_{A\subseteq[n]}|p(A)-q(A)| \label{eqn:supevent}
\end{align}
\cite[Lemma 11.6.1, p. 370]{CoverThomas2006}, the optimal-coupling form is
\begin{align}
  \TV(p,q)&=1-\sum_{i}\min(p_i,q_i) \label{eqn:optcoupling}
\end{align}
\cite[Proposition 4.2, p. 48]{LevinPeres2017}, and, since $\sum_i(p_i-q_i)=0$ forces the positive and negative parts of $p-q$ to have equal sums, the positive-part form is
\begin{align}
  \TV(p,q)&=\sum_i(p_i-q_i)^+. \label{eqn:tv-positive-part}
\end{align}

For probability vectors $p,q$ of the same length, we recall the Hellinger affinity $\rho(p,q)\triangleq\sum_i\sqrt{p_iq_i}$ (also known as the Bhattacharyya coefficient, which is a measure of \textit{similarity} between $p,q$). We also recall the KL divergence $\KL{p}{q}\triangleq\sum_{i\in\supp(p)}p_i\log(p_i/q_i)$ in nats, with convention $0\log 0=0$ and $\KL{p}{q}\triangleq+\infty$ if $\supp(p)\not\subseteq\supp(q)$. Further, we use the following notion of a specific  data processing operation. For a deterministic map $T:[n]\to[m]$, define the induced probability mass functions $p_T,q_T$ on $[m]$ by
\[
  p_T(j)=\sum_{i\in T^{-1}(j)}p_i,\qquad q_T(j)=\sum_{i\in T^{-1}(j)}q_i,\qquad j\in[m],
\]
where we take that empty sum is $0$. These need not be sorted, so in general $p_T,q_T\notin\mathcal{P}_m$. 

Throughout the rest of this section, $\kpq=\max(|\supp(p)|,|\supp(q)|)$ as in Section~\ref{sec:intro}.


\begin{lemma}[Classical facts]
\label{lem:classical}

Let $p,q$ be probability mass functions on $n$ points. Then, 
\begin{align}
  \TV(p,q)&\le\sqrt{1-\rho(p,q)^2}, \label{eqn:tv-affinity}\\
  \rho(p,q)&\ge \max\{e^{-\KL{p}{q}/2},e^{-\KL{q}{p}/2}\} \label{eqn:affinity-kl}\\
  \KL{p}{q}&\ge 2\,\TV(p,q)^2, \label{eqn:pinsker}\\
  \KL{p}{q}&\ge \log\frac{1}{1-\TV(p,q)^2}.\label{eqn:bh}
\end{align}
Also, for any map $T:[n]\to[m]$ and the associated distributions $p_T,q_T$ as defined prior to the lemma,
\begin{align}
  \KL{p}{q}&\ge\KL{p_T}{q_T}. \label{eqn:dpi}
\end{align}
\end{lemma}

Inequality \eqref{eqn:tv-affinity} is the upper Le Cam inequality and \eqref{eqn:affinity-kl} bounds the Hellinger affinity via Jensen's inequality. Both are derived as intermediate bounds in the proof of the Bretagnolle--Huber inequality \eqref{eqn:bh} \cite{BH1979} in \cite[Section 2.4]{Tsybakov2009}. Inequality \eqref{eqn:pinsker} is Pinsker's inequality \cite{Pinsker1964}, and \eqref{eqn:dpi} is a form of the data-processing inequality \cite[Theorem 2.8.1]{CoverThomas2006} which follows from \cite[Theorem 2.7.1]{CoverThomas2006}.

For $a,b\in[0,1]$ let $\dB(a,b)=a\log\frac ab+(1-a)\log\frac{1-a}{1-b}$ be the binary KL divergence where $\dB(a,b)=+\infty$ if $b\in\{0,1\}$ and $a\ne b$. 


\subsection{Comparison of TV distance and the Kolmogorov-Smirnov distance in $\Pn$}
\label{sec:kl-bounds-tightened}
In this subsection, we prove a simple relationship between the TV distance and the Kolmogorov-Smirnov (KS) distance, both of which are metrics on $\Pn.$ For $p\in \Pn$, recall the definition of the prefix sums as $S_p(l)=\sum_{i=1}^lp_i, \forall l$. 
\begin{definition}[{Prefix deviation or KS distance}]
\label{def:prefix-dev}
For $p,q\in\Pn$, the \emph{prefix deviation} is given as 
\begin{align}
  \KS(p,q):=\max_{1\le j\le \kpq}\bigl|S_p(j)-S_q(j)\bigr|, \label{eqn:ksdef}
\end{align}
which is the discrete Kolmogorov--Smirnov distance between the cumulative prefix sums \cite[Chapter 19]{VanDerVaart1998}, and is known to be a metric on $\Pn$. 
\end{definition}

We also need the following notion of single-crossing and multi-crossing pairs of distributions $p,q$. 
\begin{definition} [{Single-Crossing and Multiple-Crossing Pairs}]
    \label{def:single-crossing}
    For $p=(p_1,\hdots,p_n),q=(q_1,\hdots,q_n)\in \Pn$, let $d=p-q$, where $d_i = p_i - q_i, \forall i$. The pair $p,q\in\Pn$ is said to be \emph{single-crossing}, if there exists $j^\star\in[\kpq]$ such that either of the two conditions holds: (a) $d_i\ge 0$ for all $i\le j^\star$ and $d_i\le 0$ for all $i>j^\star$,  or (b) $d_i\le 0$ for all $i\le j^\star$ and $d_i\ge 0$ for all $i>j^\star$. A pair that is not single-crossing is said to be \emph{multiple-crossing}; equivalently, there exist $i_1<i_2<i_3$ in $[\kpq]$ with $d_{i_1}d_{i_2}<0$ and $d_{i_2}d_{i_3}<0$, i.e., the signs of the nonzero entries of $d$ change at least twice.
\end{definition}

Note that a single-crossing pair is comparable in the majorization order \cite{MOA2011}: under (a), $S_p(j)\ge S_q(j)$  for all $j$, and hence $q\preceq p$, and under (b), $S_p(j)\le S_q(j)$ for all $j$, and thus $p\preceq q$. Note also that the converse fails (e.g.\ $p=(0.5,0.2,0.2,0.1)$, $q=(0.4,0.3,0.15,0.15)$ satisfy $q\preceq p$ but are multiple-crossing). 

We now prove the following result which shows that $\KS(p,q)\le\TV(p,q)$, with equality if and only if the $p,q$ are single-crossing.

\begin{lemma}[{Total Variation Dominates Kolmogorov--Smirnov Distance}]
\label{thm:ks-vs-tv}
For all $p,q\in\Pn$, $\KS(p,q)\le\TV(p,q)$. Equality holds if and only if the pair is single-crossing.
\end{lemma}
\begin{IEEEproof}

For $j\in[\kpq]$, let $A_j=\{1,\hdots,j\}$. Observe that $\KS(p,q)=\max_{1\leq j\leq \kpq}|p(A_j)-q(A_j)|$. Hence, by \eqref{eqn:supevent}, it follows that $\KS(p,q)\le\TV(p,q)$.

We now come to the equality condition. Let $d_i:=p_i-q_i$, so that $\TV(p,q)=\sum_id_i^+$ by \eqref{eqn:tv-positive-part}. 

($\Leftarrow$) Suppose $d_i\ge0$ for $i\le j^\star$ and $d_i\le0$ for $i>j^\star$. Then $d_i=d_i^+$ for $i\le j^\star$ and $d_i^+=0$ for $i>j^\star$, so $S_p(j^\star)-S_q(j^\star)=\sum_{i\le j^\star}d_i=\sum_id_i^+=\TV(p,q)$, and hence $\KS(p,q)\ge\TV(p,q)$. By the inequality $\KS(p,q)\leq \TV(p,q)$ already shown, we see that $\KS(p,q)=\TV(p,q)$. Similarly, if $d_i\leq 0$ for $i\le j^\star$ and $d_i\geq 0$ for $i>j^\star$, then  $S_q(j^\star)-S_p(j^\star)=\TV(p,q)$, and again $\KS(p,q) = \TV(p,q)$.

($\Rightarrow$) Suppose $\KS(p,q)=\TV(p,q)$, attained at $j^\star$, i.e., $\KS(p,q)=\TV(p,q)=|p(A_{j^\star})-q(A_{j^\star})|$. Exchanging $p$ and $q$ if necessary, we have $\sum_{i\le j^\star}d_i=\TV(p,q)$, and
\begin{align}
  \TV(p,q)=\sum_{i\le j^\star}d_i\ \stackrel{(i)}{\le}\ \sum_{i\le j^\star}d_i^+\ \stackrel{(ii)}{\le}\ \sum_{i}d_i^+=\TV(p,q). \label{eqn:ks-chain}
\end{align}
Hence $(i)$ and $(ii)$ in \eqref{eqn:ks-chain} are equalities. Since $d_i^+-d_i \ge 0$ and $d_i^+ \ge 0$ for every $i$, equality in $(i)$ gives $d_i=d_i^+\ge0$ for $i\le j^\star$, and equality in $(ii)$ gives $d_i^+=0$, i.e., $d_i\le0$, for $i>j^\star$. Therefore $p,q$ are a single-crossing pair, and hence the proof is complete. 
\end{IEEEproof}

\begin{remark}
\label{rem:prefix-cuts}
For $j\in[\kpq-1]$, applying \eqref{eqn:dpi} to the map that sends $\{1,\ldots,j\}$ to one point and its complement to another gives $\KL{p}{q}\ge\dB\bigl(S_p(j),S_q(j)\bigr)$. If $p,q$ are single-crossing, then by~\Cref{thm:ks-vs-tv} some $j^\star$ satisfies $|S_p(j^\star)-S_q(j^\star)|=\TV(p,q)$. Indeed, $S_p(j^\star)=p(A_{j^\star}),$ where $A_{j^\star}=\{1,\hdots,j^\star\}$, and similarly we have $S_q(j^\star)=q(A_{j^\star}).$ This $A_{j^\star}$ is the event used in the standard proofs of \eqref{eqn:pinsker} (for instance, see \cite[Lemma~11.6.1]{CoverThomas2006}) while utilizing \eqref{eqn:supevent} for the TV distance. Thus, for single-crossing pairs, the binary divergence underlying both classical bounds can be obtained from the prefix sums alone. For pairs that are not single-crossing, no prefix can attain $\TV(p,q)$, as Lemma \ref{thm:ks-vs-tv} shows. For example, for $p=(0.5,0.3,0.2)$ and $q=(0.44,0.37,0.19)$, we have $\KS(p,q)=0.06<0.07=\TV(p,q)$. 
\end{remark}

\subsection{Refined bounds on the Hellinger affinity and total variation distance in $\Pn$}
\label{sec:affinity-bounds} \label{sec:max-tv}

The below result presents a refinement of the trivial bound on TV distance $\TV(p,q)\leq 1$, when the probability vectors are sorted. For unsorted $p,q\in\Delta_n$, $\TV=1$ is attained for disjoint supports. 
\begin{theorem}[{Maximum Total Variation Distance}]
\label{thm:maxtvd}
For any $p,q\in\Pn$, the total variation distance is bounded by: 
\begin{equation}
\TV(p,q) \le 1-\frac{1}{\kpq}, 
\label{eqn:tvdmax}
\end{equation}
with equality if and only if $\{p,q\}=\{\Unif(\kpq),\eone\}$. 
\end{theorem}
\begin{IEEEproof}
By the pigeonhole principle, $p_1 \ge 1/\kpq$ and $q_1 \ge 1/\kpq$. Therefore, $\min(p_1, q_1) \ge 1/\kpq$. Recalling (\ref{eqn:optcoupling}), $1 - \TV(p,q) = \sum_i \min(p_i, q_i) \ge \min(p_1, q_1) \ge 1/\kpq$. Rearranging gives (\ref{eqn:tvdmax}). 

If $q=\eone$ and $p=\Unif(\kpq)$, then $\sum_i\min(p_i,q_i)=\min(1/\kpq,1)=1/\kpq$, yielding exactly $\TV=1-1/\kpq$. 

Assuming equality in (\ref{eqn:tvdmax}) and substituting (\ref{eqn:optcoupling}) for $\TV$ in the LHS, we have $\sum_{i=1}^{\kpq} \min(p_i,q_i) = 1/\kpq$. Without loss of generality (by symmetry of $\TV$) say that $|\supp(p)|=\kpq$, i.e. $p$ has the larger support of the pair ($p,q$). Since $\min(p_1, q_1) \ge 1/\kpq$, tightness forces $\min(p_i,q_i)=0$ for all $i\ge 2$. Because $p_i > 0$ for all $i \in [\kpq]$, $q_i=0$ for all $i \ge 2$, establishing $q=\eone$. Tightness in the first term requires $\min(p_1,1) = 1/\kpq$, forcing $p_1 = 1/\kpq$. Because $p_1 \ge p_i$ and $\sum p_i = 1$, this necessitates $p_i=1/\kpq$ for all $i$.
\end{IEEEproof}

For general (not necessarily sorted) distributions $p,q\in\Delta_n$, the Hellinger affinity satisfies $\rho(p,q)\ge 0$, where equality is achieved when $p,q$ have disjoint supports. 
Let $\|\sqrt p\|_1:=\sum_{i=1}^{\kpq}\sqrt{p_i}$. In the following theorem, we prove a sharper lower bound on the Hellinger affinity $\rho(p,q)$, when $p,q\in \Pn$.
For this purpose, motivated by Chebyshev's sum inequality which we use in the proof, we define  $\Ch(p,q)\triangleq\frac{1}{\kpq}\,\|\sqrt p\|_1\|\sqrt q\|_1$. 
\begin{theorem}[{Bound on Hellinger affinity in $\Pn$}]
\label{thm:chebyshev}
For $p,q\in\Pn$,
\begin{align}
  \rho(p,q)\ \stackrel{(a)}{\ge}\ \Ch(p,q)\stackrel{(b)}\ge1/\kpq,\label{eqn:chebyshev-bound}
\end{align}
where inequality $(a)$ is met with equality if and only if either $p$ or $q$ is $\Unif(\kpq)$, (b) holds with equality if and only if $p=q=\eone$, in which case (a) is also an equality. 
\end{theorem}
\begin{IEEEproof}
For real sequences $x,y\in\mathbb R^{\kpq}$ which are similarly ordered (i.e., both $x,y$ are ordered in  non-increasing or non-decreasing fashion), Chebyshev's sum inequality \cite[Chapter II, Theorem 43]{HLP1934} states
\begin{align}
  \frac1{\kpq}\sum_{i=1}^{\kpq}x_iy_i\ \ge\ \Bigl(\frac1{\kpq}\sum_{i=1}^{\kpq}x_i\Bigr)\Bigl(\frac1{\kpq}\sum_{i=1}^{\kpq}y_i\Bigr). \label{eqn:chebsum}
\end{align}
Since $p$ and $q$ are non-increasing, $x_i=\sqrt{p_i}$ and $y_i=\sqrt{q_i}$ are similarly ordered and as $p_i=q_i=0$ for $i>\kpq$, $\rho(p,q)=\sum_{i=1}^{\kpq}\sqrt{p_iq_i}$. Substituting $x_iy_i=\sqrt{p_iq_i}$ in \eqref{eqn:chebsum} and multiplying by $\kpq$ gives \eqref{eqn:chebyshev-bound}.

In \cite[Theorem 43, Section 2.17]{HLP1934}, it is shown that equality in $(a)$ holds only if all the entries in either of the probability vectors are equal, i.e.\ $p=\Unif(\kpq)$ or $q=\Unif(\kpq)$. To see this directly, note the identity
\begin{align}
\label{eqn:HLPxiyiinequality}
  \nonumber \sum_{i=1}^{\kpq}&\sum_{j=1}^{\kpq}(x_i-x_j)(y_i-y_j)\\
  &=2\kpq\sum_{i=1}^{\kpq}x_iy_i-2\Bigl(\sum_{i=1}^{\kpq}x_i\Bigr)\Bigl(\sum_{i=1}^{\kpq}y_i\Bigr),
\end{align}
whose left side has nonnegative terms, so equality in \eqref{eqn:chebsum} holds if and only if $(x_i-x_j)(y_i-y_j)=0$ for all $i,j$. If $p=\Unif(\kpq)$ (or $q=\Unif(\kpq)$), then $x$ (or $y$) is constant and every term vanishes. Conversely, if neither is uniform, then, being non-increasing, $p_1>p_{\kpq}$ and $q_1>q_{\kpq}$, so $x_1>x_{\kpq}$ and $y_1>y_{\kpq}$ and the term $(i,j)=(1,\kpq)$ in \eqref{eqn:HLPxiyiinequality} is strictly positive.

For inequality $(b)$, since $\sqrt t\ge t$ on $[0,1]$, with equality if and only if $t\in\{0,1\}$, we have $\|\sqrt p\|_1\ge\sum_ip_i=1$, and likewise $\|\sqrt q\|_1\ge1$. Hence $\|\sqrt p\|_1\|\sqrt q\|_1\ge1$, and thus by definition, $\Ch(p,q)\ge1/\kpq$. Equality holds if and only if both factors equal $1$, i.e.\ every $p_i,q_i\in\{0,1\}$, which for sorted probability vectors means $p=q=\eone$. As $\kpq = 1$ for $\eone$ by definition, we have $p = q = \Unif(\kpq)$ in this case also. Hence, for above ($p,q$), equality in $(a)$ and $(b)$ holds simultaneously. 
\end{IEEEproof}

The following remark shows that the lower bound $\Ch(p,q)$ for $\rho(p,q)$ in Theorem \ref{thm:chebyshev} dominates the bound \eqref{eqn:affinity-kl} in Lemma \ref{lem:classical} for some $p,q$ choices.
\begin{remark}[{Neither affinity bound dominates}]
\label{prop:jensen-cheb}
Let $\mathcal J(p,q):=\max\{e^{-\KL{p}{q}/2},e^{-\KL{q}{p}/2}\}$. 
We now show that there exist $p,q\in\Pn$ with $\mathcal J(p,q)>\Ch(p,q)$ as well as $p,q\in\Pn$ with $\Ch(p,q)>\mathcal J(p,q)$. Indeed, 
for $p=q=(0.7,0.3)$, $\KL{p}{q}=0$ gives $\mathcal J=1$, whereas $\Ch(p,q)=\tfrac12(\sqrt{0.7}+\sqrt{0.3})^2\approx0.958$. For $p=(0.5,0.5)$ and $q=(0.9,0.1)$, by Theorem \ref{thm:chebyshev} we have $\Ch(p,q)=\rho(p,q)=\sqrt{0.45}+\sqrt{0.05}\approx0.894$, while $\KL{q}{p}\approx0.368$ and $\KL{p}{q}\approx0.511$ give $\mathcal J\approx0.832$.
\end{remark}


The proof for the following corollary is straightforward by combining \eqref{eqn:tv-affinity} with Theorem \ref{thm:chebyshev}. 
\begin{corollary}[{A $\Ch(p,q)$-based bound for $\TV(p,q)$ in $\Pn$}]
\label{thm:sbh}
For $p,q\in\Pn$, 
\begin{align}
  \TV(p,q)\ \le\ \sqrt{1-\Ch(p,q)^2}. \label{eqn:sbh} 
\end{align}
\end{corollary}

\begin{remark} [{Comparison of the bound from \Cref{thm:sbh} with pre-existing bounds for $\TV(p,q)$ distance}]
Rearranging \eqref{eqn:pinsker} and \eqref{eqn:bh}, and by symmetry, we get $\TV(p,q)\le\min\left(\sqrt{\KL{p}{q}/2},\sqrt{\KL{q}{p}/2}\right)$ and $\TV(p,q)\leq \min\left(\sqrt{1-e^{-\KL{p}{q}}},\sqrt{1-e^{-\KL{q}{p}}}\right)$, respectively. The bound in Corollary \ref{thm:sbh} is smaller than $\min\left(\sqrt{1-e^{-\KL{p}{q}}},\sqrt{1-e^{-\KL{q}{p}}}\right)$ whenever $\Ch(p,q)^2>\max \left(e^{-\KL{p}{q}}, e^{-\KL{q}{p}}\right)$, which occurs in some cases by \Cref{prop:jensen-cheb}. 

Neither the bound $\sqrt{1-\Ch(p,q)^2}$ nor Pinsker's bound $\sqrt{\min\{\KL{p}{q},\KL{q}{p}\}/2}$ for $\TV(p,q)$ dominates the other as well. For $p=\eone$ and $q=\Unif(6)$ in $\mathcal P_6$, we have $\TV(p,q)=5/6$, $\KL{q}{p}=+\infty$ and $\KL{p}{q}=\log6$, so Pinsker's inequality gives $\TV(p,q)\le\sqrt{(\log 6)/2}\approx0.947$, whereas $\Ch(p,q)=\frac16\cdot1\cdot\sqrt6=1/\sqrt6$ gives $\TV(p,q)\le\sqrt{5/6}\approx0.913$. Hence, there exist $p,q\in\Pn$ such that the $\Ch(p,q)$-based bound is smaller than the Pinsker bound. Conversely, for $p=q=(0.7,0.3)$ in $\mathcal P_2$, Pinsker's inequality gives the exact value $\TV(p,q)=0$, whereas $\Ch(p,q)=\frac12(\sqrt{0.7}+\sqrt{0.3})^2\approx0.958$ gives only $\TV(p,q)\le\sqrt{1-\Ch(p,q)^2}\approx0.286$. Note that the bound from \Cref{thm:maxtvd} is exactly the value of the TV distance for this example, that is $\frac{5}{6}$.
\end{remark}

The following example compares the bound $\sqrt{1-\Ch(p,q)^2}$ of Corollary~\ref{thm:sbh} with the bound $1-1/\kpq$ of Theorem~\ref{thm:maxtvd}, and shows that neither is always smaller than the other.
\begin{example}
\label{ex:tv-bounds-neither}
Let $n=2$ and $p=(1/2,1/2)$, so that $\kpq=2$ and $1-1/\kpq=1/2$ for every $q\in\mathcal P_2$. First, let $q=\eone=(1,0)$. Then $\TV(p,q)=1/2$, so the bound $1-1/\kpq$ of Theorem~\ref{thm:maxtvd} is attained, whereas $\Ch(p,q)=\frac12\cdot\sqrt2\cdot1=1/\sqrt2$ gives $\sqrt{1-\Ch(p,q)^2}=1/\sqrt2>1/2$. Next, let $q=(9/16,7/16)$. Then $\TV(p,q)=1/16$ and $\Ch(p,q)=\frac12\cdot\sqrt2\cdot\frac{3+\sqrt7}{4}$, so $\Ch(p,q)^2=(8+3\sqrt7)/16$ and $\sqrt{1-\Ch(p,q)^2}=\sqrt{8-3\sqrt7}/4\approx0.0626$. This is smaller than $1/2$, since $3\sqrt7>4$, and is within $2\times10^{-4}$ of $\TV(p,q)$, whereas $1-1/\kpq=1/2$. Hence neither of the bounds $1-1/\kpq$ and $\sqrt{1-\Ch(p,q)^2}$ dominates the other. 
\end{example}

\section{Conclusion}

In summary, this work is an attempt to expand the landscape of (dis)similarity measures on the majorization simplex. Relationships between these measures, improved bounds, connections to minimum-entropy coupling and functional representations, are currently being investigated. Applications in ecology and econometrics are also viable areas for further research. 

\newpage 
\bibliographystyle{IEEEtran}
\bibliography{IEEEabrv,ITW26}


\newpage 

\appendices

\section{Proof of \Cref{thm:metric}}
\label{app:metric-proof}

\begin{IEEEproof}

We verify the four metric axioms. 

\emph{Non-negativity:} Since $p\preceq p\vee q$ and $q\preceq p\vee q$, Schur-concavity of $H_\alpha$ gives $H_\alpha(p)\geq H_\alpha(p\vee q)$ and $H_\alpha(q)\geq H_\alpha(p\vee q)$. Adding these gives $d_\alpha(p,q) = H_\alpha(p)+H_\alpha(q)-2H_\alpha(p\vee q) \geq 0$.

\emph{Identity of indiscernibles:} If $p=q$ then $p\vee q=p$, so $d_\alpha(p,p)=2H_\alpha(p)-2H_\alpha(p)=0$. Conversely, suppose $d_\alpha(p,q)=0$. By non-negativity, both $H_\alpha(p)-H_\alpha(p\vee q)\geq0$ and $H_\alpha(q)-H_\alpha(p\vee q)\geq0$, and their sum is zero, so each vanishes. For every finite $\alpha\in[1,\infty)$, $H_\alpha$ is strictly Schur-concave, so $p\preceq p\vee q$ together with equal entropies forces $p=p\vee q$; symmetrically $q=p\vee q$. Hence $p=q$.

\emph{Symmetry:} $d_\alpha(p,q)=d_\alpha(q,p)$ is immediate from commutativity of the join, $p\vee q=q\vee p$.

\emph{Triangle inequality:} To show $d_\alpha(p,q)\leq d_\alpha(p,r)+d_\alpha(r,q)$, expand and cancel to obtain the equivalent inequality $H_\alpha(p\vee r) + H_\alpha(r\vee q) - H_\alpha(r) \leq H_\alpha(p\vee q)$. Since $p\preceq p\vee r$ and $q\preceq r\vee q$, both $p$ and $q$ lie below $(p\vee r)\vee(r\vee q)$ in $\Pn$. Now, $p\vee q$ is the least distribution in $\Pn$ amongst any distribution which majorizes $p$ and $q$. We thus have $p\vee q\preceq (p\vee r)\vee(r\vee q)$, so by Schur-concavity $H_\alpha(p\vee q) \geq H_\alpha\bigl((p\vee r)\vee(r\vee q)\bigr)$. Applying ~\Cref{thm:supermod} to the pair $(p\vee r,\,r\vee q)$ gives $H_\alpha(p\vee r) + H_\alpha(r\vee q) \leq H_\alpha\bigl((p\vee r)\wedge(r\vee q)\bigr) + H_\alpha\bigl((p\vee r)\vee(r\vee q)\bigr)$. Since $r$ is a lower bound for both $p\vee r$ and $r\vee q$, we have $r\preceq(p\vee r)\wedge(r\vee q)$, so by Schur-concavity $H_\alpha\bigl((p\vee r)\wedge(r\vee q)\bigr) \leq H_\alpha(r)$. Combining the arguments in this paragraph gives $H_\alpha(p\vee r)+H_\alpha(r\vee q)-H_\alpha(r) \leq H_\alpha\bigl((p\vee r)\vee(r\vee q)\bigr)$, establishing the required inequality.

\end{IEEEproof}

\section{Remark: Why $\alpha\in[0,1)\cup\{\infty\}$ fails to give a metric} 
\label{app:da-metric-range}
The extensions of Definition~\ref{def:dalpha} to $\alpha=0$ and $\alpha=\infty$ fail the condition $d_\alpha(p,q)=0\implies p=q$ (for $p,q\in\Pn$, general $n$). For instance, for $p=(0.5,0.3,0.2)\neq(0.5,0.4,0.1)=q$ in $\mathcal P_3$, we have $H_\infty(p)=H_\infty(q)=-\log 0.5$, and since $p\vee q$ also has largest entry $0.5$, we thus have $d_\infty(p,q)=0$. Similarly, for $p=(0.9,0.1,0)\neq(0.7,0.3,0)=q$, both have support size $2$ and $p\vee q=p$, so $d_0(p,q)=0$. This occurs because $H_0$ and $H_\infty$ are not strictly Schur-concave, which is required for showing the property that $d_\alpha(p,q)=0\Rightarrow p=q$. However, $d_0$ and $d_\infty$ satisfy other properties, and hence act as pseudometrics in $\Pn$, for general $n$. Note that \cite[Theorem 4]{YS26} shows that supermodularity fails for $\alpha\in(0,1)$ and that we use supermodularity to prove triangle inequality in Appendix \ref{app:metric-proof}. Thus, our proof of triangle inequality breaks down for $\alpha\in(0,1)$. Indeed, we show an instance for which the triangle inequality does not hold for $\alpha=1/2$. Thus, $H_{1/2}(v)=2\log\sum_i\sqrt{v_i}$ and, for $p,q\in\mathcal P_3$,
\begin{align*}
  d_{1/2}(p,q)=2\log\frac{h(p)\,h(q)}{h(p\vee q)^2},\qquad h(v):=\sum_i\sqrt{v_i}.
\end{align*}
Take
\begin{align*}
  p&=\Bigl(\tfrac{1089}{1225},\tfrac{100}{1225},\tfrac{36}{1225}\Bigr),\\
  q&=\Bigl(\tfrac{784}{1225},\tfrac{441}{1225},0\Bigr),\\
  r&=\Bigl(\tfrac{900}{1225},\tfrac{324}{1225},\tfrac{1}{1225}\Bigr),
\end{align*}
all in $\mathcal P_3$. Their cumulative sums are $S_p=\frac{1}{1225}(1089,1189,1225)$, $S_q=\frac{1}{1225}(784,1225,1225)$ and $S_r=\frac{1}{1225}(900,1224,1225)$. The joins have cumulative sums equal to the coordinatewise maxima of these, and these maxima are already concave, so
\begin{align*}
  p\vee q&=\tfrac{1}{1225}(1089,136,0),\\
  p\vee r&=\tfrac{1}{1225}(1089,135,1),\\
  r\vee q&=\tfrac{1}{1225}(900,325,0).
\end{align*}
Hence
\begin{align*}
  h(p)=h(q)&=h(r)=\tfrac75,\\
  h(p\vee q)&=\tfrac{33+2\sqrt{34}}{35},\\
  h(p\vee r)&=\tfrac{34+3\sqrt{15}}{35}, \\ h(r\vee q)&=\tfrac{30+5\sqrt{13}}{35}.
\end{align*}
Since $h(p)=h(q)=h(r)$, the inequality $d_{1/2}(p,q)>d_{1/2}(p,r)+d_{1/2}(r,q)$ is equivalent to
\[
  h(p\vee r)\,h(r\vee q)\;>\;h(r)\,h(p\vee q),
\]
equivalently,
\[
  (34+3\sqrt{15})(30+5\sqrt{13})\;>\;49\,(33+2\sqrt{34})=1617+98\sqrt{34}.
\]
Using $\sqrt{15}>3.872$ and $\sqrt{13}>3.605$ (since $3.872^2<15$ and $3.605^2<13$), the left side exceeds $45.616\cdot 48.025>2190$. Using $\sqrt{34}<5.831$ (since $5.831^2>34$), the right side is less than $1617+98\cdot 5.831<2189$. Hence $d_{1/2}(p,q)>d_{1/2}(p,r)+d_{1/2}(r,q)$.

\section{{Proof of \Cref{lem:comparability} ($\Rightarrow$)}}
\label{app:comparability-proof}

We use the extension of $\preceq$ to $\mathbb R^N$ and the notation $\mathsf{sort}(\cdot\,\|\,\cdot)$ from Section~\ref{sec:setup}. 

\begin{lemma}[{Hardy--Littlewood--P\'olya inequality \cite{HLP1934}, \cite[Chapter 3, Sec. C]{MOA2011}}]
\label{lem:hlp}
Let $\phi:[0,1]\to\mathbb R$ be convex and continuous, and let $a,b\in[0,1]^N$ satisfy $a\preceq b$. Then
\[
  \sum_{i=1}^N\phi(a_i)\ \le\ \sum_{i=1}^N\phi(b_i).
\]
\end{lemma}

The majorization precursor on the lattice is the following relation, which we recall in the present notation.

\begin{lemma}[{Majorization precursor, \cite[Thm.~1]{SJDC26}}]
\label{lem:precursor}
For all $x,y\in\Pn$, in $\mathbb R^{2n}$,
\[
  \mathsf{sort}(x\wedge y\,\|\,x\vee y)\ \preceq\ \mathsf{sort}(x\,\|\,y).
\]
\end{lemma}


\begin{IEEEproof}[Proof of \Cref{lem:comparability} ($\Rightarrow$)]
Suppose $x,y\in\Pn$ are incomparable and $\alpha\in(1,\infty)$. We show that the supermodularity inequality of Theorem~\ref{thm:supermod} is strict, i.e., $H_\alpha(x)+H_\alpha(y)<H_\alpha(x\wedge y)+H_\alpha(x\vee y)$. Write $F(p):=\sum_ip_i^\alpha$ and
\[
  \widehat A=F(x),\quad \widehat B=F(y),\quad \widehat M=F(x\wedge y),\quad \widehat J=F(x\vee y).
\]
Since $H_\alpha=\tfrac1{1-\alpha}\log F$ with $\tfrac1{1-\alpha}<0$ and $\log$ increasing, the claim is equivalent to
\begin{align}
\label{eqn:Fprod}
  \widehat A\,\widehat B\ >\ \widehat M\,\widehat J .
\end{align}

The map $t\mapsto t^\alpha$ is strictly convex on $[0,1]$ for $\alpha>1$, so $F$ is strictly Schur-convex on $\Pn$ (equivalently, $H_\alpha$ is strictly Schur-concave, Definition~\ref{def:schur}). Incomparability means $x\wedge y\notin\{x,y\}$ and $x\vee y\notin\{x,y\}$, since $x\wedge y=x\iff x\preceq y$ and likewise for the join. As $x\wedge y\preceq x,y\preceq x\vee y$, we obtain, strictly,
\begin{align}
\label{eqn:strictorder}
  \widehat M < \widehat A,\quad \widehat M<\widehat B, \qquad
  \widehat J > \widehat A,\quad \widehat J>\widehat B .
\end{align}

Next we apply Lemma~\ref{lem:hlp} to the precursor of Lemma~\ref{lem:precursor}. Take $N=2n$, $\phi(t)=t^\alpha$, which is convex and differentiable on $[0,1]$, and
\[
  a=\mathsf{sort}(x\wedge y\,\|\,x\vee y),\qquad b=\mathsf{sort}(x\,\|\,y).
\]
Lemma~\ref{lem:precursor} states $a\preceq b$. Sorting does not change the sum of $\phi$ over the entries, so $\sum_i\phi(a_i)=F(x\wedge y)+F(x\vee y)=\widehat M+\widehat J$ and $\sum_i\phi(b_i)=F(x)+F(y)=\widehat A+\widehat B$. Lemma~\ref{lem:hlp} therefore gives
\begin{align}
\label{eqn:additive}
  \widehat A+\widehat B \ \ge\ \widehat M+\widehat J .
\end{align}

It remains to combine \eqref{eqn:strictorder} and \eqref{eqn:additive}. Let $x_0:=\widehat A-\widehat M$ and $y_0:=\widehat J-\widehat B$. Both are positive by \eqref{eqn:strictorder}, and \eqref{eqn:additive} reads $x_0\ge y_0$. Then
\begin{align}
  \widehat A\widehat B-\widehat M\widehat J
  &= \widehat A\widehat B-(\widehat A-x_0)(\widehat B+y_0) \nonumber\\
  &= x_0(\widehat B+y_0)-\widehat Ay_0
   = x_0\widehat J-\widehat Ay_0 \nonumber\\
  &\stackrel{\text{(A)}}{\ge} y_0\widehat J-\widehat Ay_0
   = y_0(\widehat J-\widehat A)
  \stackrel{\text{(B)}}{>} 0 .
  \label{eqn:supermodnonstrictness}
\end{align}
Here (A) holds by \eqref{eqn:additive}, which is $x_0\ge y_0$, together with $\widehat J>0$; and (B) holds by \eqref{eqn:strictorder}, which gives $y_0>0$ and $\widehat J>\widehat A$. Thus \eqref{eqn:supermodnonstrictness} establishes \eqref{eqn:Fprod}, which proves the claim.
\end{IEEEproof}

\begin{remark}[A note on Topkis's lemma]
    A positive function $f$ on a lattice $(L,\preceq)$ is called \emph{log-submodular} if $f(x)f(y)\ge f(x\wedge y)f(x\vee y)$ for all $x,y\in L$ (strictly log-submodular if this inequality is strict). Topkis \cite[Lemma~2.6.6(b)]{Topkis1998} shows that every positive, increasing, submodular function on a lattice is log-submodular. However, no condition for strictness of log-submodularity is given. For $F(p)=\sum_ip_i^\alpha$ with $\alpha>1$, which is positive, increasing and submodular \eqref{eqn:additive} on $\Pn$, Topkis's lemma yields $F(x)F(y)\ge F(x\wedge y)F(x\vee y)$, which is the non-strict supermodularity of R\'enyi entropy $H_\alpha=\tfrac1{1-\alpha}\log F$ in Theorem~\ref{thm:supermod}. Thus, another way to look at Lemma \ref{lem:comparability} is that it gives the conditions for strictness in R\'enyi supermodularity.
\end{remark}

\subsection{Issue with the approach of \cite{SJDC26} regarding the equality conditions for supermodularity of Shannon and Tsallis entropies}

\begin{remark}[An issue with \cite{SJDC26}]
\label{rem:sjdc-strictness}
In \cite{SJDC26}, strict supermodularity of the Shannon and Tsallis entropies on incomparable pairs of distributions is derived from the claim that a certain `precursor' relation holds, i.e., $\mathsf{sort}(x\wedge y\,\|\,x\vee y)\preceq\mathsf{sort}(x\,\|\,y)$ is never an equality when $x$ and $y$ are incomparable. Example~\ref{ex:touch-not-comparable} shows that equality in the precursor relation can occur for an incomparable pair, so this argument does not establish strictness of supermodularity of Shannon and Tsallis entropies for every incomparable pair. For R\'enyi entropy of order $\alpha>1$ the supermodularity inequality is strict for every incomparable pair (Lemma~\ref{lem:comparability}), by an argument that does not pass through equality in the precursor. 
\end{remark}

\begin{example}[Tight Shannon supermodularity for an incomparable pair]
\label{ex:touch-not-comparable}
Let $n=4$, $p=(0.45,0.25,0.2,0.1)$ and $q=(0.4,0.3,0.25,0.05)$. Then $S_p=(0.45,0.7,0.9,1)$ and $S_q=(0.4,0.7,0.95,1)$, so $S_p(1)>S_q(1)$ but $S_p(3)<S_q(3)$ and the pair is incomparable. Now, $p\wedge q=(0.4,0.3,0.2,0.1)$, $p\vee q=(0.45,0.25,0.25,0.05)$, and
\begin{align}
\label{eqn:prec}
  \mathsf{sort}(p\,\|\,q)&=(0.45,0.4,0.3,0.25,0.25,0.2,0.1,0.05) \nonumber
  \\&=\mathsf{sort}(p\wedge q\,\|\,p\vee q).
\end{align}
Now, applying \Cref{lem:hlp} with $\phi(x) = x\log(x)$, $N = 2n$, $a = \sort{p||q)}$, $b = \sort{p\wedge q||p\vee q)}$ to \eqref{eqn:prec} (where $a = b$), we get $H_1(p)+H_1(q)=H_1(p\wedge q)+H_1(p\vee q)$. For $\alpha=2$, $\sum_ip_i^2=\sum_iq_i^2=0.315$ while $\sum_i(p\wedge q)_i^2=0.3$ and $\sum_i(p\vee q)_i^2=0.33$, so $0.315^2=0.099225>0.099=0.3\times0.33$, i.e.\ $H_2(p)+H_2(q)<H_2(p\wedge q)+H_2(p\vee q)$, as Lemma~\ref{lem:comparability} predicts.
\end{example}

\end{document}

%% file: macros.tex
\newcommand{\suppress}[1]{}

\newcommand{\supp}{\mathsf{supp}}  

